\documentclass[a4paper]{amsart}
\usepackage[utf8]{inputenc}
\usepackage{amsfonts}
\usepackage{amsmath,amssymb,graphicx,epsfig}
\usepackage{enumerate,amsthm,epstopdf}
\usepackage{bbm}
\usepackage{dsfont}
\usepackage{algpseudocode}
\usepackage{algorithmicx}
\usepackage{algorithm}
\usepackage{listings}
\usepackage{hyperref}
\hypersetup{colorlinks,linkcolor={blue},citecolor={bred},urlcolor={blue}}
\usepackage{cleveref}
\usepackage{xcolor}

\definecolor{bred}{rgb}{0.8,0,0}
\usepackage{fullpage}
\usepackage[utf8]{inputenc}
\usepackage[T1]{fontenc}
\usepackage[english]{babel}

\usepackage[makeroom]{cancel}
\usepackage{verbatim}

\DeclareMathOperator{\var}{var}

\newcommand{\argmin}{\operatornamewithlimits{argmin}}

\def\cF{{\mathcal F}}
\def\cG{{\mathcal G}}
\def\cH{{\mathcal H}}

\def\cP{{\mathcal P}}
\def\sX{{\mathsf X}}

\def\bR{{\mathbb R}}

\def\bE{{\mathbb E}}

\def\bN{{\mathbb N}}

\def\NPDF{{\mathcal N}}

\def\f0{{\mathbf 0}}

\def\md{{\mathrm{d}}}

\def\qed{$\Box$}

\newtheorem{thm}{Theorem}

\newtheorem{assumption}{Assumption}

\newtheorem{prop}{Proposition}

\theoremstyle{definition}

\newtheorem{rem}{Remark}

\title{Global convergence of \\ optimized adaptive importance samplers}

\author{\"Omer Deniz Akyildiz}
\address{Department of Mathematics, Imperial College London, UK}
\email{{\textcolor{bred}{\footnotesize \texttt{deniz.akyildiz@imperial.ac.uk}}}}

\subjclass[2020]{Primary 65C05, 65C35; 65D30; 62F99; 90C26}

\keywords{Adaptive importance sampling, variance minimization, non-convex optimization}  
  
\begin{document}
\maketitle
\begin{abstract}
We analyze the \textit{optimized adaptive importance sampler} (OAIS) for performing Monte Carlo integration with general proposals. We leverage a classical result which shows that the bias and the mean-squared error (MSE) of the importance sampling scales with the $\chi^2$-divergence between the target and the proposal and develop a scheme which performs global optimization of $\chi^2$-divergence. While it is known that this quantity is convex for exponential family proposals, the case of the general proposals has been an open problem. We close this gap by utilizing the nonasymptotic bounds for stochastic gradient Langevin dynamics (SGLD) for the global optimization of $\chi^2$-divergence and derive nonasymptotic bounds for the MSE by leveraging recent results from non-convex optimization literature. The resulting AIS schemes have explicit theoretical guarantees that are uniform-in-time.
\end{abstract}

\section{Introduction}\label{secIntro}
Importance sampling (IS) is one of the most fundamental methods to compute expectations w.r.t. a target distribution $\pi$ using samples from a proposal distribution $q$ and reweighting these samples. This procedure is known to be inefficient when the discrepancy between $\pi$ and $q$ is large. To remedy this, adaptive importance samplers (AIS) are based on the principle that one can iteratively update a sequence of proposal distributions $(q_k)_{k\geq 1}$ to obtain refined and better proposals over time. This provides a significant improvement over a naive importance sampler with a single proposal $q$. For this reason, AIS schemes received a significant attention over the past decades and enjoy an ongoing popularity, see, e.g., \cite{bengio2008adaptive,bugallo2015adaptive,martino2015adaptive, kappen2016adaptive, bugallo2017adaptive,elvira2017improving,martino2017layered,elvira2019generalized}. The general and most generic AIS scheme retains $N$ distinct distributions centred at the samples from the previous iteration and constructs a mixture proposal; variants of this approach include population Monte Carlo (PMC) \cite{cappe2004population} or adaptive mixture importance sampling \cite{cappe2008adaptive}. Although these versions of the methods have been widely popular, all these methods still lack theoretical guarantees and convergence results as the number of iterations grows to infinity (see \cite{douc2007convergence} for an analysis in terms of $N$). In other words, there has been a lack of theoretical guarantees about whether this kind of adaptation moves the proposal density towards the target, and if so, in which metric and at what rate. The difficulty of providing such rates stems from the fact that it is difficult to quantify the convergence of the nonparametric mixture distributions to the target measure.

In this paper, we provide an analysis to address this fundamental question for a different (and more tractable) class of samplers, parametric AIS schemes, using the available results from nonconvex optimization literature. Recently, this fundamental theoretical problem was addressed by \cite{akyildiz2021convergence} who considered a specific family of proposals, i.e., the exponential family as fixed proposal family. In this case, a fundamental quantity in the MSE bound of the importance sampler, specifically the $\chi^2$-divergence (or equivalently the variance of the importance weights), can be shown to be convex which leads to a natural adaptation strategy based on convex optimization, see, e.g., \cite{arouna2004adaptative,arouna2004robbins, kawai2008adaptive, lapeyre2011framework, ryu2014adaptive, kawai2017acceleration, kawai2018optimizing} for the algorithmic applications of this property (see also \cite{fu2002optimal} for an application in a financial context). This quantity appeared and was investigated in other contexts, e.g., sequential Monte Carlo methods \cite{cornebise2008adaptive}, asymptotic analysis \cite{delyon2018asymptotic}, or to determine the necessary sample size for the IS \cite{sanz2018importance, sanz2021bayesian}. The convexity property of $\chi^2$-divergence when the proposal is from the exponential family was exploited by \cite{akyildiz2021convergence} to prove finite error bounds for the AIS that are {\textit{uniform-in-time}}, in particular, providing a general convergence rate $\mathcal{O}(1/\sqrt{k}N + 1/N)$ for the $L_2$ error for the importance sampler, where $k$ is the number of iterations and $N$ is the number of Monte Carlo samples used for integration. A similar result was extended to adaptive optimisers in \cite{perello2023adaptively}. However, these results do not apply for a general proposal distribution, as this results in a function in the MSE bound that is non-convex in the parameter of the proposal.

We address the problem of analzing the optimized AIS in the general setting by applying non-convex optimization results for $\chi^2$ divergence. This enables us to prove global convergence results for the AIS that can be controlled by the parameters of the non-convex optimization schemes. Specifically, we use stochastic gradient Langevin dynamics (SGLD) \cite{welling2011bayesian} for the analysis of the non-convex optimization schemes which minimize $\chi^2$-divergence. SGLD is a common proxy for the analysis of SGD and often exhibit similar properties, see, e.g., \cite{brosse2018promises}. Recently, global convergence of these algorithms for non-convex optimization were shown in several works, see, e.g., \cite{raginsky2017non, xu2018global, erdogdu2018global, zhang2023nonasymptotic, akyildiz2020nonasymptotic, gao2021global, lim2021polygonal, lim2021non}.

{We note that the use of {Markov chain Monte Carlo based proposals} within the AIS is explored before \cite{martino2017layered,martino2017anti,llorente2021mcmc}. In particular, the Langevin algorithm based proposals have been also used, see, e.g., \cite{fasiolo2018langevin,elvira2019langevin,elvira2021optimized,mousavi2021hamiltonian, elvira2022gradient}. However, these ideas are distinct from our work, in the sense that they explore driving the mixture parameters w.r.t. the gradient of the log-target, i.e., $\nabla_x \log\pi(x)$ and using these parameters to construct mixture proposals in the standard AIS setting. We instead use Langevin algorithm to minimise $\chi^2$-divergence. Our proposal adaptation approach is motivated by quantitative error bounds, hence has provable guarantees. Other MCMC-based methods also perform well and are interesting for a future analysis -- but require a different approach. Finally, various other measures to optimize proposals were also considered in the literature such as the Kullback-Leibler divergence or other parametric measures, see, e.g., \cite{el2019variational, el2019stochastic, el2021policy, jerfel2021variational}. These methods also perform well in practice but are harder to verify from a theoretical perspective.}

\textbf{Organization.} The paper is organized as follows. In Sec~\ref{sec:background}, we provide a brief background of adaptive importance sampling schemes and, specifically, the parametric AIS which we aim at analyzing. We also introduce the fundamental results on which we rely in later sections. In Section~\ref{sec:algorithms}, we describe the setting to globally optimize the $\chi^2$-divergence between the target and the proposal. In Section~\ref{sec:analysis}, we prove nonasymptotic error rates for the MSE for the setting where the adaptation is driven by standard gradient methods using results from non-convex optimization literature. These bounds are then discussed in detail in Section~\ref{sec:discussion}. Finally, we conclude with Section~\ref{sec:conclusions}.

\subsection*{Notation}
For an integer $k\in\bN$, we denote $[k] = \{1,\ldots,k\}$. The state-space is denoted as $\sX$ where $\sX \subseteq \bR^{d_x}$ with $d_x \geq 1$. We use $B(\sX)$ to denote the set of bounded functions on $\sX$, and $\cP(\sX)$ to denote the set of probability measures on $\sX$, respectively. We write $(\varphi,\pi) = \int \varphi(x) \pi(\mbox{d}x)$ or $\bE_\pi[\varphi(X)]$ and $\var_\pi(\varphi) = (\varphi^2,\pi) - (\varphi,\pi)^2$. {The notation $\delta_x(\mathrm{d} y)$ denotes a Dirac measure centred at $x$, i.e., for $\varphi \in B(\sX)$, we have $\varphi(x) = \int \varphi(y) \delta_x(\mathrm{d}y)$.}

We will use $\pi$ to denote the target distribution. Accordingly, we use $\Pi$ to denote the unnormalized target, i.e., we have $\pi(x) = \Pi(x)/Z_{\pi}$. We denote the proposal distribution with $q_\theta$ where $\theta \in \bR^{d_\theta}$ where $d_\theta$ denotes the parameter dimension. We denote both the measures, $\pi$ and $q_\theta$, and their densities with the same letters.

To denote the minimum value of functions $\rho, R$, we use $\rho_\star, R_\star$.

\section{Background}\label{sec:background}
In this section, we give a brief background and formulation of the problem. 
\subsection{Importance sampling}
Given a target density $\pi \in \cP(\sX)$, we are interested in computing integrals of the form
\begin{align}\label{eq:ProbDefn}
(\varphi,\pi) = \int_\sX \varphi(x) \pi(x) \mbox{d}x.
\end{align}
We assume that we can only evaluate the unnormalized density and cannot sample from $\pi$ directly. Importance sampling is based on the idea of using a proposal distribution to sample from and weight these samples to account for the discrepancy between the target and the proposal. These weights and samples are finally used to construct an estimator of the integral. In particular, let $q_\theta\in\cP(\sX)$ be the proposal with parameter $\theta \in \bR^{d_\theta}$, then the unnormalised target density $\Pi:\sX \to \bR_+$ {satisfies}
\begin{align*}
\pi(x) = \frac{\Pi(x)}{Z_\pi},
\end{align*}
where $Z_\pi :=\int_\sX \Pi(x) \md x < \infty$. Next, we define the unnormalized weight function $W_\theta:\sX \times \bR^{d_\theta} \to \bR_+$ as
\begin{align*}
W_\theta(x) = \frac{\Pi(x)}{q_\theta(x)}.
\end{align*}
Given a target $\pi$ and a proposal $q_\theta$, the importance sampling procedure first draws a set of independent and identically distributed (iid) samples $\{x^{(i)}\}_{i=1}^N$ from $q_\theta$. Next, we construct the empirical measure $\pi^N_\theta$ as
\begin{align*}
\pi_\theta^N(\mbox{d}x) = \sum_{i=1}^N \mathsf{w}_\theta^{(i)} \delta_{x^{(i)}}(\mbox{d}x),
\end{align*}
where,
\begin{align*}
\mathsf{w}_\theta^{(i)} = \frac{W_\theta(x^{(i)})}{\sum_{j=1}^N W_\theta(x^{(j)})}.
\end{align*}
Finally this measure yields the self-normalizing importance sampling (SNIS) estimate
\begin{align}\label{eq:SNISestimate}
(\varphi,\pi^N_\theta) = \sum_{i=1}^N \mathsf{w}_\theta^{(i)} \varphi(x^{(i)}).
\end{align}
Although the estimator \eqref{eq:SNISestimate} is biased in general, one can show that the bias and the MSE vanish with a rate $\mathcal{O}(1/N)$. Below, we present the well-known MSE bound (see, e.g., \cite{agapiou2017importance} or \cite{akyildiz2021convergence}).
\begin{thm}\label{thm:ISfund}
Assume that $(W_\theta^2,q_\theta) < \infty$. Then for any $\varphi\in B(\sX)$, we have
\begin{align}\label{eq:MSEbound}
\bE\left[\left((\varphi,\pi) - (\varphi,\pi_{\theta}^N)\right)^2\right] \leq \frac{c_\varphi \rho(\theta)}{{N}},
\end{align}
where $c_\varphi = 4\|\varphi\|_\infty^2$ and the function $\rho:\Theta \to [\rho_\star,\infty)$ is defined as
\begin{align}
\rho(\theta) = \bE_{q_\theta}\left[\frac{\pi^2(X)}{q^2_\theta(X)}\right],
\label{eqThm1-2}
\end{align}
where $\rho_\star := \inf_{\theta\in\Theta} \rho(\theta) \geq 1$.
\end{thm}
\begin{proof}
See \cite[Theorem~2.1]{agapiou2017importance} or \cite[Theorem~1]{akyildiz2021convergence} for a proof.
\end{proof}
\begin{rem}\label{rem:UnnormalizedBound}
It will be useful for us to write the bound \eqref{eq:MSEbound} as
\begin{align}\label{eq:UnnormalizedBound}
\bE\left[\left((\varphi,\pi) - (\varphi,\pi_{\theta}^N)\right)^2\right] \leq \frac{c_\varphi R(\theta)}{{N Z_\pi^2}},
\end{align}
where
\begin{align}\label{eq:Rtheta}
R(\theta) = \bE_{q_\theta}\left[\frac{\Pi^2(X)}{q_\theta^2(X)}\right].
\end{align}
Note that while the function $\rho$ and related quantities (such as its gradients) cannot be computed by sampling from $q_\theta$ (since we cannot evaluate $\pi(x)$), same quantities for $R(\theta)$ can be computed since $\Pi(x)$ can be evaluated.
\end{rem}
\begin{rem}\label{rem:rhoToChi} As shown in \cite{agapiou2017importance}, the function $\rho$ can be written in terms of $\chi^2$ divergence between $\pi$ and $q_\theta$, i.e.,
\begin{align*}
\rho(\theta) := \chi^2(\pi || q_\theta) + 1.
\end{align*}
Note also that $\rho(\theta)$ can also be written in terms of the variance of the weight function $w_\theta = \pi(x) / q_\theta(x)$, which \textit{is} the $\chi^2$-divergence, i.e.,
\begin{align*}
\rho(\theta) = \var_{q_\theta}(w_\theta(X)) + 1.
\end{align*}
\end{rem}
Finally, a similar result can be presented for the bias from \cite{agapiou2017importance}.
\begin{thm}\label{thm:SNISbias}
Assume that $(W_\theta^2,q_\theta) < \infty$. Then for any $\varphi\in B(\sX)$, we have
\begin{align*}
\left| \bE\left[(\varphi,\pi_{\theta}^N)\right] - (\varphi,\pi) \right| \leq \frac{\bar{c}_\varphi \rho(\theta)}{{N}},
\end{align*}
where $\bar{c}_\varphi = 12\|\varphi\|_\infty^2$ and the function $\rho:\Theta \to [\rho_\star,\infty)$ is the same as in Theorem~\ref{thm:ISfund}.
\end{thm}
\begin{proof}
See Theorem~2.1 in \cite{agapiou2017importance}.
\end{proof}

\subsection{Parametric adaptive importance samplers}
Importance sampling schemes tend to perform poorly in practice when the chosen proposal is ``far away'' from the target -- leading to samples with degenerate weights, i.e., {most of the importance weights become zero}. {This reduces the sampler inefficiency, as in effect, this results in fewer samples, which can be measured by Effective Sample Size (ESS) \cite{elvira2022rethinking}.} We can already see this fact from Theorem~\ref{thm:ISfund}: For any parametric family $q_\theta$, the function $\rho(\theta)$ defines a distance measure between $\pi$ and $q_\theta$. A large discrepancy between the target and the proposal implies a large $\rho$, which degrades the error bound. For this reason, in practice, the proposals are \textit{adapted}, meaning that they are refined over iterations to better match the target. In literature, mainly, the adaptive mixture {proposals} are employed, see, e.g., \cite{cappe2004population, bugallo2017adaptive} and many variants including multiple proposals are proposed, see, e.g., \cite{martino2017layered,elvira2019generalized}.

In contrast to the mixture samplers, we review here the \textit{parametric AIS}. In this scheme, the proposal distribution is not a mixture with weights, but instead, a parametric family of distributions, denoted $q_\theta$. Adaptation, therefore, becomes a problem of updating the parameter $\theta^\eta_k$, where $\eta$ is the parameter of the updating mechanism, which results in a sequence of proposal distributions denoted $(q_{\theta^\eta_k})_{k\geq 1}$.

\begin{algorithm}[t]
\begin{algorithmic}[1]
\caption{Parametric AIS}\label{alg:ParametricAIS}
\State Choose a parametric proposal $q_{\theta}$ with initial parameter $\theta=\theta_0$.
\For{$t\geq 1$}
\State Adapt the proposal,
\begin{align*}
\theta^\eta_k = \mathcal{T}_{k,\eta}(\theta^\eta_{k-1}),
\end{align*}
\State Sample,
\begin{align*}
x_k^{(i)} \sim q_{\theta^\eta_k}, \quad \textnormal{for } i = 1,\ldots,N,
\end{align*}
\State Compute weights,
\begin{align*}
\mathsf{w}_{\theta^\eta_k}^{(i)} = \frac{W_{\theta^\eta_k}(x_k^{(i)})}{\sum_{i=1}^N W_{\theta^\eta_k}(x_k^{(i)})}, \textnormal{ where }W_{\theta^\eta_k}^{(i)} = \frac{\Pi(x_k^{(i)})}{q_{\theta^\eta_k}(x^{(i)})}.
\end{align*}
\State Report the point-mass probability measure
\begin{align*}
{\pi}_{\theta^\eta_k}^N(\md x) = \sum_{i=1}^N \mathsf{w}_{\theta^\eta_k}^{(i)} \delta_{x_k^{(i)}}(\md x),
\end{align*}
and the estimator
\begin{align*}
(\varphi,{\pi}_{\theta^\eta_k}^N) = \sum_{i=1}^N \mathsf{w}_{\theta^\eta_k}^{(i)} \varphi(x_k^{(i)}).
\end{align*}
\EndFor
\end{algorithmic}
\end{algorithm}

Consider the proposal distribution $q_{\theta^\eta_{k-1}}$ at iteration $k-1$. For performing one step of this scheme, the parameter $\theta^\eta_{k-1}$ is updated via a mapping
\begin{align*}
\theta^\eta_k = \mathcal{T}_{\eta,k}(\theta^\eta_{k-1}),
\end{align*}
where $\{\mathcal{T}_{\eta,k}:\Theta \to \Theta, k\geq 1\}$, is a sequence of deterministic or stochastic maps parameterized by $\eta$, typically in the form of optimizers (hence $\eta$ can be the step-size). We then continue with the conventional importance sampling technique, by simulating from this proposal
\begin{align*}
x_k^{(i)} \sim q_{\theta^\eta_k}(\md x), \quad \textnormal{for } i = 1,\ldots,N,
\end{align*}
computing the weights
\begin{align*}
\mathsf{w}_{\theta^\eta_k}^{(i)} = \frac{W_{\theta^\eta_k}(x_k^{(i)})}{\sum_{i=1}^N W_{\theta^\eta_k}(x_k^{(i)})},
\end{align*}
and finally constructing the empirical measure
\begin{align*}
\pi_{\theta^\eta_k}^N(\md x) = \sum_{i=1}^N \mathsf{w}_{\theta^\eta_k}^{(i)} \delta_{x_k^{(i)}}(\md x).
\end{align*}
The estimator of the integral \eqref{eq:ProbDefn} can be computed as in Eq. \eqref{eq:SNISestimate}. 

The parametric AIS method is given in Algorithm~\ref{alg:ParametricAIS}. We can now adapt Theorem~\ref{thm:ISfund} to this particular, time-varying case.
\begin{thm}\label{thm:ISfundAIS}
Assume that, given a sequence of proposals $(q_{\theta^\eta_k})_{k\geq 1} \in \cP(\sX)$, we have $(W_{\theta^\eta_k}^2,q_{\theta^\eta_k}) < \infty$ for every $k \geq 1$. Then for any $\varphi\in B(\sX)$, we have
\begin{align*}
\bE\left[\left|(\varphi,\pi) - (\varphi,\pi_{\theta^\eta_k}^N)\right|^2\right] \leq \frac{c_\varphi \rho(\theta^\eta_k)}{{N}},
\end{align*}
where $c_\varphi = 4 \|\varphi\|_\infty^2$ and the function $\rho(\theta^\eta_k):\Theta \to [\rho_\star,\infty)$ is defined as in Eq. \eqref{eqThm1-2}.
\end{thm}
\begin{proof}
The proof is identical to the proof of Theorem~\ref{thm:ISfund}. We have just re-stated the result to introduce the iteration index $k$.
\end{proof}
This result is useful in the sense of providing a finite error bound, however, it does not indicate whether iterations of the AIS help reducing the error. 

\subsection{Adaptation as global nonconvex optimization}
When $q_\theta$ is an exponential family density, it is shown that $\rho(\theta)$, and consequently $R(\theta)$, are convex functions \cite{ryu2014adaptive,ryu2016convex,akyildiz2021convergence}. Based on this, \cite{akyildiz2021convergence} have proved convergence rates for stochastic gradient based adaptation algorithms which minimize $\rho$ and $R$ (i.e. the variance of the importance weights) assuming an exponential family $q_\theta$. They proved finite-time uniform MSE bounds since convex optimization algorithms have well-known convergence rates. In particular, they showed that the optimized AIS with stochastic gradient descent as the minimization procedure has $\mathcal{O}(1/\sqrt{k}N + 1/N)$ convergence rate which vanishes as $k$ and $N$ grows. While this rate is first of its kind for adaptive importance samplers, it has been limited to a single proposal family (the exponential family). In general, when $q_\theta$ is not from exponential family, then $\rho$ and $R$ are non-convex functions.

In this paper, we do not limit the choice of $q_\theta$ to any fixed proposal family. Therefore, in the adaptation step, we are interested in solving the global nonconvex optimization problem
\begin{align*}
\theta^\star \in \argmin_{\theta \in \bR^{d_\theta}} R(\theta),
\end{align*}
where $R(\theta)$ is given in \eqref{eq:Rtheta}. This will lead to a global optimizer $\theta^\star$ which will give the best possible proposal in terms of minimizing the MSE of the importance sampler. We utilize convergence results of stochastic gradient Langevin dynamics (SGLD) \cite{zhang2023nonasymptotic} for the analysis of the global optimization properties of such methods. We summarize the setting in the next section.
\section{The Setting}\label{sec:algorithms}
In this section, we describe the algorithmic setting we analyze. It is important to note that we do not propose a new algorithm, but rather aim at investigating the behaviour of the algorithms minimizing the variance of the importance weights ($\chi^2$-divergence) \cite{arouna2004adaptative,arouna2004robbins, kawai2008adaptive, lapeyre2011framework, ryu2014adaptive, kawai2017acceleration, kawai2018optimizing} in the setting where the proposal is not exponential family. One natural example is the setting of \cite{muller2019neural} where the authors parameterized the proposal with a neural network and minimized $\chi^2$-divergence to choose a proposal.

We note that, within this section, we only consider the case of self-normalized importance sampling (SNIS) which is the practical case. We also assume, we have only stochastic estimates of the gradient of the $R(\theta)$ function.

\begin{rem} The gradient can be computed as
\begin{align*}
\nabla R(\theta) &= \nabla_\theta \int \frac{\Pi^2(x)}{q_\theta^2(x)} q_\theta(x) \md x = \nabla_\theta \int \frac{\Pi^2(x)}{q_\theta(x)} \md x = - \int \frac{\Pi^2(x)}{q_\theta^2(x)} \nabla q_\theta(x) \md x = - \int \frac{\Pi^2(x)}{q_\theta^2(x)} \nabla \log q_\theta(x) q_\theta(x) \md x,
\end{align*}
which leads to
\begin{align}\label{eq:gradR}
\nabla R(\theta) = - \bE_{q_\theta}\left[\frac{\Pi^2(X)}{q_\theta^2(X)} \nabla \log q_\theta(X)\right].
\end{align}
Therefore the stochastic estimate of $\nabla R(\theta)$ can be obtained by sampling from $q_\theta$, a straightforward and routine operation of the AIS. We also remark that this gradient can be written in terms of the unnormalized weight function
\begin{align*}
\nabla R(\theta) = - \bE_{q_\theta}\left[W^2_\theta(X) \nabla \log q_\theta(X)\right].
\end{align*}
This suggests that the adaptation will use weights and samples from $q_\theta$, which makes this operation much closer to the classical mixture AIS approaches.
\end{rem}

\begin{algorithm}[t]
{\begin{algorithmic}[1]
\caption{Parametric AIS with SGLD}\label{alg:ParametricAIS_specialised}
\State Choose a parametric proposal $q_{\theta}$ with initial parameter $\theta=\theta_0$.
\For{$t\geq 1$}
\State Adapt the proposal,
\begin{align*}
\theta^\eta_k = \theta^\eta_{k-1} + \frac{\Pi^2(g_{\theta^\eta_{k-1}}(\varepsilon_k))}{q_{\theta^\eta_{k-1}}^2(g_{\theta^\eta_{k-1}}(\varepsilon_k))} \nabla \log q_{\theta^\eta_{k-1}}(g_{\theta^\eta_{k-1}}(\varepsilon_k)) \nabla g_{\theta^\eta_{k-1}}(\varepsilon_k) + \sqrt{2 \eta} W_k,
\end{align*}
where $\varepsilon_k \sim r_\varepsilon$ and $W_k \sim \mathcal{N}(0, I_{d_\theta})$.
\State Sample,
\begin{align*}
x_k^{(i)} \sim q_{\theta^\eta_k}, \quad \textnormal{for } i = 1,\ldots,N,
\end{align*}
\State Compute weights,
\begin{align*}
\mathsf{w}_{\theta^\eta_k}^{(i)} = \frac{W_{\theta^\eta_k}(x_k^{(i)})}{\sum_{i=1}^N W_{\theta^\eta_k}(x_k^{(i)})}, \textnormal{ where }W_{\theta^\eta_k}^{(i)} = \frac{\Pi(x_k^{(i)})}{q_{\theta^\eta_k}(x^{(i)})}.
\end{align*}
\State Report the point-mass probability measure
\begin{align*}
{\pi}_{\theta^\eta_k}^N(\md x) = \sum_{i=1}^N \mathsf{w}_{\theta^\eta_k}^{(i)} \delta_{x_k^{(i)}}(\md x),
\end{align*}
and the estimator
\begin{align*}
(\varphi,{\pi}_{\theta^\eta_k}^N) = \sum_{i=1}^N \mathsf{w}_{\theta^\eta_k}^{(i)} \varphi(x_k^{(i)}).
\end{align*}
\EndFor
\end{algorithmic}}
\end{algorithm}

\subsection{Low Variance Gradient Estimation}
We assume that the proposal is reparameterizable: We assume $x \sim q_\theta$ can be performed by first sampling $\varepsilon \sim r_\varepsilon$ and setting $x = g_\theta(\varepsilon)$. Therefore, the gradient expression in eq.~\eqref{eq:gradR} becomes
\begin{align*}
\nabla R(\theta) = - \bE_{r_\varepsilon}\left[\frac{\Pi^2(g_\theta(\varepsilon))}{q_\theta^2(g_\theta(\varepsilon))} \nabla \log q_\theta(g_\theta(\varepsilon)) \nabla g_\theta(\varepsilon) \right].
\end{align*}
{This is a common way to estimate the gradient and it typically reduces the variance of the gradient estimate \cite{xu2019variance}. In the above expression, $r_\varepsilon$ is a distribution that is independent of $\theta$. A simple example of the reparameterization can be given as follows. Let $q_\theta(x) = \mathcal{N}(x; \theta, \sigma^2)$. Then we can choose, e.g., $r_\varepsilon(\varepsilon) = \mathcal{N}(\varepsilon; 0, 1)$ and $g_\theta(\varepsilon) = \theta + \sigma \varepsilon$. This generalizes to any location-scale family (e.g. Student's t-distribution, see Section~\ref{sec:experiment}).}

We remark that this does not limit the flexibility of our parametric family, as reparameterization is widely used as a variance reduction technique in variational inference (VI) and variational autoenconders (VAEs) and a flexible choice of parametric families is possible via this mechanism (see \cite{dieng2017variational} and \cite{lopez2020decision} for applications of $\chi^2$-divergence minimization in VI and VAEs, respectively). A second motivation to do so is to consider the numerical difficulties related to high-variance in estimating $\chi^2$-divergence as laid out by \cite{pradier2019challenges}. Finally, the Langevin dynamics with stochastic gradients is well studied when the randomness in the gradient is independent of the parameter of interest. It is therefore natural to consider this setting for gradient estimation.

We denote the stochastic gradient accordingly as $H(\theta, \varepsilon)$ and define
\begin{align}\label{eq:StochasticGradient}
H(\theta, \varepsilon) = - \frac{\Pi^2(g_\theta(\varepsilon))}{q_\theta^2(g_\theta(\varepsilon))} \nabla \log q_\theta(g_\theta(\varepsilon)) \nabla g_\theta(\varepsilon).
\end{align}
In order to prove convergence of the schemes we analyze, we assume certain regularity conditions of this term, see Section~\ref{sec:analysis} for details.
\subsection{Global optimization of AIS}
The implementation of the minimization procedure of the variance of importance weights or $\chi^2$-divergence can be done using gradient-based techniques and this is widespread in literature \cite{arouna2004adaptative,arouna2004robbins, kawai2008adaptive, lapeyre2011framework, ryu2014adaptive, kawai2017acceleration, kawai2018optimizing}. We describe below one scheme to model the behaviour of such algorithms (most notably, stochastic gradient descent) under the setting where the proposal can be outside the exponential family. We note that SGLD is a common way to analyze the properties of SGD and displays similar behaviour in certain settings \cite{brosse2018promises}.

Consider the problem of minimizing $R(\theta)$ with a gradient-based algorithm We can most generally consider stochastic gradient Langevin dynamics (SGLD) \cite{welling2011bayesian,zhang2023nonasymptotic} as a good model for such a scheme to adapt the proposal. For this purpose, we consider the mappings $\mathcal{T}_{\eta,k}$ as SGLD steps
\begin{align}\label{eq:SGLDadapt}
\theta^\eta_{k+1} = \theta_{k}^\eta - \eta H(\theta_{k}^\eta, \varepsilon_k) + \sqrt{\frac{2\eta}{\beta}} \xi_{k+1},
\end{align}
i.e., $\mathcal{T}_{\eta,k}(\theta_k^\eta) = \theta_{k}^\eta - \eta H(\theta_{k}^\eta, \varepsilon_k) + \sqrt{\frac{2\eta}{\beta}} \xi_{k+1}$, where $\varepsilon_k \sim r_\varepsilon$, $\bE[H(\theta, \varepsilon_k)] = \nabla R(\theta)$, and $(\xi_k)_{k\in\bN}$ are {multivariate Normals with zero mean and identity covariance}. The parameter $\beta$ is called the inverse temperature parameter. Note that we consider a single sample estimate of the gradient $\nabla R(\theta)$ as it is customary in the gradient estimation literature with reparameterization trick. This mapping $\mathcal{T}_{\eta,k}$ acts as a global optimizer in Algorithm~\ref{alg:ParametricAIS} as we described before.

\section{Analysis}\label{sec:analysis}
In this section, we provide the analysis of the adaptive importance samplers described above. The main argument in our analysis relies on the fact that SGLD recursion in \eqref{eq:SGLDadapt} (and in general Langevin dynamcis) can be seen as \textit{global optimizers} \cite{zhang2023nonasymptotic}. In particular, a recursion of type \eqref{eq:SGLDadapt} converges to a target measure of the form $\pi_\beta \propto \exp(-\beta R(\theta))$ which, as $\beta \to \infty$, concentrates on the minimizers of $R(\theta)$ \cite{hwang1980laplace}. Therefore, one can see these samplers as global optimizers as with large $\beta$, the samples will be arbitrarily close to minima.

In particular, we start by assuming that the adaptation can be driven by an exact gradient $\nabla R(\theta)$ as an illustrative case and analyze this case in Section~\ref{sec:analysis:deterministic}. Albeit unrealistic, this gives us a starting point. Then we analyze the case where the adaptation is driven by the SGLD in Section~\ref{sec:analysis:SOLAIS}.

\subsection{Convergence rates for deterministic gradient case}\label{sec:analysis:deterministic}
In this section, we provide a simplified analysis to give the intuition of our main results. This case considers a fictitious scenario where the gradients of $R$ can be exactly obtained. This algorithm would correspond to an exact gradient descent on $\chi^2$-divergence with injected noise. This gives us a clear picture about the role of noise (often comes from stochastic gradients) in adaptation. To do so, consider the overdamped Langevin dynamics to optimize the parameters of the proposal
\begin{align}\label{eq:DeterministicLangevin}
\theta_{k+1}^\eta = \theta_k^\eta - \eta \nabla R(\theta_k^\eta) + \sqrt{\frac{2 \eta}{\beta}} \xi_{k+1}.
\end{align}
While without the addition of noise variables $(\xi_{k})_{k\geq 0}$, this gradient-based method would only converge to a local minimum, we will show below that the addition of noise will lead to a global minimization algorithm In order to do so, we place the following assumptions on $R$.
\begin{assumption}\label{assmp:deterministicLip} The gradient of $R$ is $L_R$-Lipschitz, i.e., for any $\theta, \theta' \in \bR^d$,
\begin{align}
\|\nabla R(\theta) - \nabla R(\theta')\| \leq L_R \|\theta - \theta'\|.
\end{align}
\end{assumption}
Next, we assume the standard dissipativity assumption in non-convex optimization literature.
\begin{assumption}\label{assmp:deterministicDissip} The gradient of $R$ is $(m_R, b_R)$-dissipative, i.e., for any $\theta$
\begin{align}
\left\langle \nabla R(\theta), \theta \right\rangle \geq m_R \|\theta\|^2 - b_R.
\end{align}
\end{assumption}
It is worth discussing that, instead of placing assumptions on the proposal family $q_\theta$, as done in many prior works, we place the assumptions on $\chi^2$-divergence directly. This places \textit{implicit} assumptions on the proposal, but since these conditions are relaxed, this covers a much larger class of proposals than exponential family.

We can now adapt Theorem~3.3 of \cite{xu2018global} in order to understand optimization properties of the recursion \eqref{eq:DeterministicLangevin}. In summary, the next theorem shows that the recusion \eqref{eq:DeterministicLangevin} acts as a global optimizer.
\begin{thm}\label{thm:33ofXu} \cite[Theorem 3.3]{xu2018global} Under Assumptions \ref{assmp:deterministicLip}-\ref{assmp:deterministicDissip}, we obtain
\begin{align*}
\bE[R(\theta_k^\eta)] - R_\star \leq c_1 e^{-c_0 k \eta} + \frac{c_2}{\beta} \eta + c_3,
\end{align*}
where
\begin{align}\label{eq:deterministic_c_3}
c_3 = \frac{d}{2\beta} \log \left( \frac{e L_R (b_R \beta / d + 1)}{m_R} \right),
\end{align}
where $R_\star = \min_{\theta\in\bR^d} R(\theta)$ and $c_0, c_1, c_2 > 0$ are constants given in \cite[Theorem~3.3]{xu2018global}.
\end{thm}
In order to shed light onto some of the intuition, we note that $c_0$ is related to the spectral gap of the underlying Markov chain, characterizing the speed of convergence of the underlying continuous-time Langevin diffusion to the target. The constant $c_2$ is a result of the discretization error of the Langevin algorithm Finally, $c_3$ is the error caused by the fact that the latest sample of the Markov chain $\theta_k^\eta$ is used to estimate the optima, i.e., $c_3$ quantifies the gap between
\begin{align*}
\bE[R(\theta_\infty)] - R_\star,
\end{align*}
where $\theta_\infty \sim \pi_\beta \propto \exp(- \beta R(\theta))$, i.e., a random variable with the target measure $\pi_\beta$ of the chain. This gap is independent of $\eta$.

We next provide the MSE result of the importance sampler whose proposal is driven by the Langevin algorithm \eqref{eq:DeterministicLangevin}.
\begin{thm}\label{thm:LangevinAIS} Let Assumptions \ref{assmp:deterministicLip} and \ref{assmp:deterministicDissip} hold, let $(\theta_k^\eta)_{k\geq 1}$ be generated by the recursion in \eqref{eq:DeterministicLangevin}, and assume that for a sequence of proposals $(q_{\theta_k^\eta})_{k\geq 1} \in \cP(\sX)$, we have $(W_{\theta_k^\eta}, q_{\theta_k^\eta}) < \infty$ for every $k$. Then for any $\varphi \in B(\sX)$, we have
\begin{align}
\bE\left[\left| (\varphi, \pi) - (\varphi, \pi_{\theta_k^\eta}^N)\right|^2\right] &\leq \frac{c_{\varphi,\pi} c_1 e^{- c_0 k \eta}}{N} + \frac{c_2 c_{\varphi, \pi}}{\beta} \frac{\eta}{N} + \frac{c_{\varphi, \pi} c_3}{N} + \frac{c_{\varphi} \rho_\star}{N}.\label{eq:MSELangevin}
\end{align}
where $c_{\varphi,\pi} = c_\varphi / Z_\pi^2$ and $c_0, c_1,c_2,c_3$ are given in Theorem~\ref{thm:33ofXu}.
\end{thm}
\begin{proof}
Let $\cF_{k-1} = \sigma(\theta_0^\eta, \ldots, \theta_{k-1}^\eta)$. We note using Theorem~\ref{thm:ISfundAIS}, we have
\begin{align*}
\bE\left[\left| (\varphi, \pi) - (\varphi, \pi_{\theta_k^\eta}^N)\right|^2 \Big\vert \cF_{k-1}  \right] &\leq \frac{c_\varphi R(\theta_k^\eta)}{Z_\pi^2 N}, \\
&\leq \frac{c_\varphi (R(\theta_k^\eta) - R_\star)}{Z_\pi^2 N} + \frac{c_\varphi \rho_\star}{N}.
\end{align*}
Taking expectations of both sides and using Theorem~\ref{thm:33ofXu} for the first term on the r.h.s. concludes the result.
\end{proof}
This result provides a {\textit{uniform-in-time error}} bound for the adaptive importance samplers with general proposals.
\subsection{Convergence rates of SGLD adapted proposals}\label{sec:analysis:SOLAIS}
In this section, we start with placing assumptions on stochastic gradients $H(\theta, \varepsilon)$ as defined in \eqref{eq:StochasticGradient}. We note that these assumptions are the most relaxed conditions to prove the convergence of Langevin dynamics to this date, see, e.g., \cite{zhang2023nonasymptotic,chau2021stochastic}. We first need to assume that sufficient moments of the distribution $r_\varepsilon$ exists.
\begin{assumption}\label{assmp:moments} We have $|\theta_0| \in L^4$. The process $(\varepsilon_k)_{k\in\bN}$ is i.i.d. with $|\varepsilon_0|\in L^{4(\rho+1)}$. Also, $\bE[H(\theta, \varepsilon_0)] = \nabla R(\theta)$.
\end{assumption}
Next, we place a \textit{local} Lipschitz assumption on $H$.
\begin{assumption}\label{assump:LocalLip} There exists positive constants $L_1$, $L_2$, and $\rho$ such that
\begin{align*}
(i) \quad |H(\theta, \varepsilon) - H(\theta',\varepsilon)| &\leq L_1 (1 + |\varepsilon|)^\rho |\theta - \theta'| \\
(ii) \quad |H(\theta, \varepsilon) - H(\theta,\varepsilon')| &\leq L_2 (1 + |\varepsilon| + |\varepsilon'|)^\rho (1 + |\theta|) |\varepsilon - \varepsilon'|
\end{align*}
\end{assumption}
Finally, we assume a local dissipativity assumption.
\begin{assumption}\label{assmp:localDisp} There exist $M : \bR^{d_\varepsilon} \to \bR^{d_\theta \times d_\theta}$, $b: \bR^{d_\varepsilon} \to \bR$ such that for any $x,y \in \bR^{d_\theta}$,
\begin{align*}
\left\langle y, M(x)y\right\rangle \geq 0
\end{align*}
and for all $\theta \in \bR^{d_\theta}$ and $\varepsilon \in \bR^{d_\varepsilon}$,
\begin{align*}
\left\langle H(\theta,\varepsilon), \theta\right\rangle \geq \left\langle \theta, M(\varepsilon) \theta\right\rangle - b(\varepsilon).
\end{align*}
\end{assumption}
\begin{rem} We note that we can relate parameters introduced in these assumptions to the ones we introduced in the deterministic case $L_R$ and $b_R$. In particular,
\begin{align*}
L_R = L_1 \bE[(1 + |\varepsilon_0|)^\rho], \quad \textnormal{and} \quad b_R = \bE[b(\varepsilon_0)].
\end{align*}
We also note that the smallest eigenvalue of the matrix $\bE[M(\varepsilon_0)]$ is $m_R$.
\end{rem}
We remark again that these assumptions are now placed on the \textit{stochastic gradient} $H(\theta,\varepsilon)$ instead of the proposal $q_\theta$ as in \cite{akyildiz2021convergence}. As such, they are implicit.

We can finally state the convergence result of the SGLD for non-convex optimization from \cite{zhang2023nonasymptotic}.
\begin{thm}\label{thm:SGLD} \cite[Corollary~2.9]{zhang2023nonasymptotic}  Let $\theta_k^\eta$ be generated by the SGLD recursion \eqref{eq:SGLDadapt}. Let Assumptions \ref{assmp:moments}, \ref{assump:LocalLip}, and \ref{assmp:localDisp} hold. Then, there exist constants $c_0, c_1, c_2, c_3, \eta_{max} > 0$ such that for every $0<\eta\leq \eta_{\max}$,
\begin{align*}
\bE[R(\theta_k^\eta)] - R_\star \leq c_1 e^{-c_0 \eta k} + c_2 \eta^{1/4} + c_3,
\end{align*}
where $c_0, c_1, c_2, c_3, \eta_{\max}$ are given explicitly in \cite{zhang2023nonasymptotic}.
\end{thm}
This is a similar result to the case of deterministic gradients but covers the realistic scenario of stochastic gradients. One can see the effect of stochasticity of gradients in the rate of convergence, i.e., while the case of deterministic gradients has a rate $\mathcal{O}(\eta)$, SGLD can only guarantee a rate of $\mathcal{O}(\eta^{1/4})$ (which is known to be suboptimal).

With this result at hand, we can state the global convergence result of SGLD-driven AIS.
\begin{thm}\label{thm:SOLAIS} Let $\theta_k^\eta$ be generated by the SGLD recursion \eqref{eq:SGLDadapt}. Let Assumptions \ref{assmp:moments}, \ref{assump:LocalLip}, and \ref{assmp:localDisp} hold. Then
\begin{align*}
\bE\left[\left| (\varphi, \pi) - (\varphi, \pi_{\theta_k^\eta}^N) \right|^2\right] &\leq \frac{c_{\varphi, \pi} c_1 e^{-c_0 \eta k}}{N} + \frac{c_2 c_{\varphi, \pi} \eta^{1/4}}{N} + \frac{c_3 c_{\varphi, \pi}}{N} + \frac{c_\varphi \rho_\star}{N}.
\end{align*}
\end{thm}
\begin{proof}
Let $\cF_{k-1} = \sigma(\theta_0^\eta, \varepsilon_1, \ldots, \varepsilon_{k-1})$ and $\cG_k = \sigma(\xi_1,\ldots,\xi_k)$. Let $\cH_k = \cF_{k-1} \bigvee \cG_k$. We next note
\begin{align*}
\bE\left[\left| (\varphi, \pi) - (\varphi, \pi_{\theta_k^\eta}^N) \right|^2 \Big\vert \cH_{k}\right] \leq \frac{c_{\varphi, \pi} R(\theta_k^\eta)}{N}.
\end{align*}
We expand the r.h.s. as
\begin{align*}
\frac{c_{\varphi, \pi} R(\theta_k^\eta)}{N} =c_{\varphi, \pi} \frac{R(\theta_k^\eta) - R_\star}{N} + \frac{c_\varphi \rho_\star}{N}.
\end{align*}
Taking unconditional expectations of boths sides, we obtain
\begin{align*}
\bE\left[\left| (\varphi, \pi) - (\varphi, \pi_{\theta_k^\eta}^N) \right|^2\right] \leq & \, c_{\varphi, \pi} \frac{\bE[R(\theta_k^\eta)] - R_\star}{N} + \frac{c_\varphi \rho_\star}{N}.
\end{align*}
Using Theorem~\ref{thm:SGLD} for the term $\bE[R(\theta_k^\eta)] - R_\star$, we obtain the result.
\end{proof}
We can again see that this is a uniform-in-time result for the AIS. As opposed to Theorem~\ref{thm:LangevinAIS}, the dependence to step-size in this theorem is worse: It is $\mathcal{O}(\eta^{1/4})$ rather than $\mathcal{O}(\eta)$. The difference between this result and Theorem~\ref{thm:LangevinAIS} about the deterministic case is twofold: First, we assume that the gradients are stochastic, which is the case for real applications. Second, for the stochastic gradient $H(\theta, \varepsilon)$, our assumptions are the weakest possible assumptions, hence allows us to choose a wider family. It is possible, for example, to obtain better dependence in $\eta$ if one assumes that stochastic gradients are uniformly Lipschitz, see, e.g., \cite{xu2018global}.

\section{Discussion}\label{sec:discussion}
In this section, we provide a discussion of our results.

\subsection{Discussion of the constants in error bounds}
In this section, we summarize and discuss the constants in error bounds to provide intuition about the utility of our results. We restrict our attention to SGLD-driven AIS (i.e. we do not consider the deterministic scheme). In our discussion, we use $c_0, c_1, c_2, c_3$ to denote constants in Theorem~\ref{thm:SOLAIS}.

\vphantom{x}

\noindent\textbf{Dimension dependence.} Because dissipative non-convex potentials can cover worst case scenarios, the dimension dependence of $c_1, c_2$ are $\mathcal{O}(e^d)$ and $c_0 = \mathcal{O}(e^{-d})$ \cite{zhang2023nonasymptotic,akyildiz2020nonasymptotic}. These bounds are, however, worst case scenarios and reflect the edge cases. In practice, SGLD performs well with non-convex potentials, leading to well-performing methods. Recall that $c_3$ is given by
\begin{align}\label{eq:c_3_disc}
c_3 = \frac{d}{2\beta} \log \left( \frac{e L_R (b_R \beta / d + 1)}{m_R} \right).
\end{align}
In this case, one can see that $c_3 = \mathcal{O}(d \log(1/d))$, which degrades the bound as $d$ grows.

\vphantom{x}

\noindent\textbf{Dependence of inverse temperature $\beta$}. We note that $c_0 $, $c_1$, and $c_2$ are $\mathcal{O}(1/\beta)$ whereas $\beta$-dependence of $c_3$ is $\mathcal{O}(\log \beta / \beta)$ as can be seen from \eqref{eq:c_3_disc}. This suggests a strategy to set $\beta$ large enough so that $c_3 = \mathcal{O}(\log \beta / \beta) \leq \epsilon$ to vanish $c_3$ from the bound. If this is satisfied, then the second term $c_2 \eta^{1/4}$ can be controlled by the step-size and the first term $c_1 e^{-c_0 \eta k}$ vanishes as $k\to\infty$.

\vphantom{x}

\noindent\textbf{Calibrating step-sizes and the number of particles.} The discussion also suggests a possible heuristic to calibrate the step-sizes and the number of particles of the method: For sufficiently large $k$ (so that the first term in \eqref{eq:MSELangevin} is sufficiently small), setting $N = \eta^{-\alpha}$ with $\alpha > 0$ provides an overall MSE bound 
\begin{align}\label{eq:MSELangevinAlpha}
\bE\left[\left| (\varphi, \pi) - (\varphi, \pi_{\theta_k^\eta}^N)\right|^2\right] \leq \mathcal{O}(\eta^\alpha).
\end{align}
Therefore, one can trade computational efficiency with the statistical accuracy of the method as manifested by our error bound. For example, a small $\alpha$ would correspond to a low number of particles, but a potentially high MSE.

{\section{Numerical Example}\label{sec:experiment} In this section, we construct a heavy-tailed proposal family $q_\theta$ and use it to discuss Assumptions~\ref{assump:LocalLip} and \ref{assmp:localDisp} and demonstrate the idea numerically. Since it is known that proposals within the exponential family lead to a convex optimization problem \cite{akyildiz2021convergence}, we construct a proposal family that is outside the exponential family.}

\begin{figure}[t]
    \includegraphics[width=\textwidth]{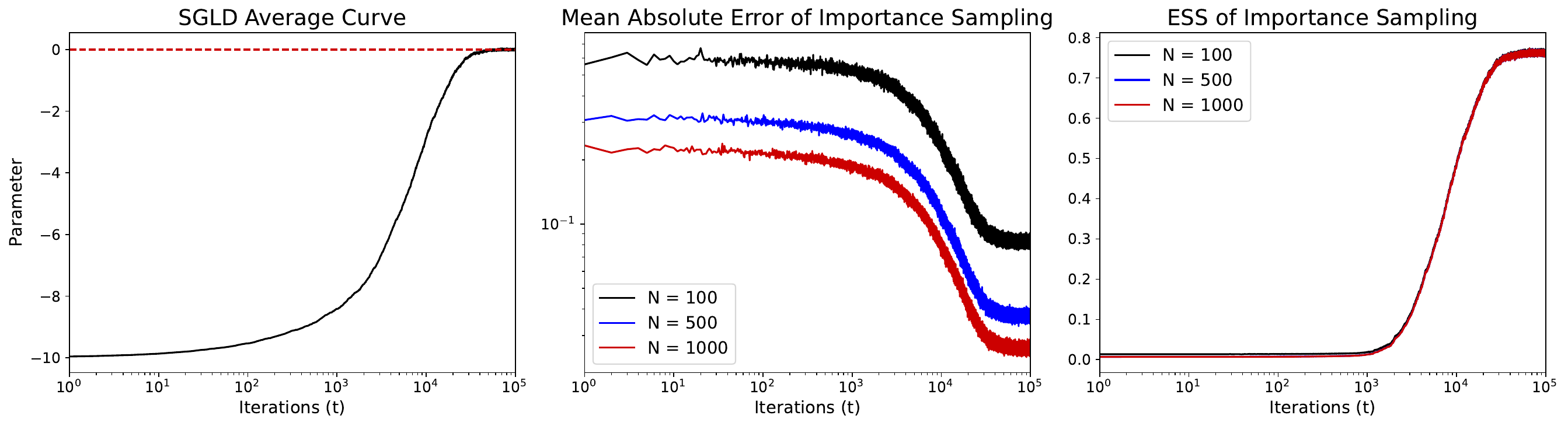}
    \caption{Implementation of SGLD for the Gaussian target with a Student's-t proposal. The results are averaged over $1000$ runs.}
    \label{fig:sgld}
\end{figure}
{For this, let $q_\theta$ be the Student's t distribution with fixed scale $\tau > 0$ and $1 < \nu < 2$ degrees of freedom. We parameterise its mean by $\theta$ and obtain
\begin{align*}
q_\theta(x) = \frac{\Gamma\left(\frac{\nu + 1}{2}\right)}{\Gamma\left(\frac{\nu}{2}\right)} \frac{1}{\sqrt{\pi \nu \tau^2}} \left(1 + \frac{(x - \theta)^2}{\nu \tau^2}\right)^{-\frac{\nu + 1}{2}}.
\end{align*}
We also note the general form
\begin{align}
    \partial_\theta \log q_\theta(x) := \frac{\partial \log q_\theta}{\partial \theta} =\frac{(\nu + 1) (x - \theta)}{\tau^2 \nu + {(x - \theta)^2}}.
\end{align}
Given a sample from a centered $\varepsilon \sim t_\nu(0, 1)$ distribution, we can obtain a sample from $q_\theta$ by adding $\theta$ to it, i.e., $g_\theta(\varepsilon) = \tau \varepsilon + \theta$. For simplicity, let $\pi(x)$ be a Gaussian density:
\begin{align*}
\pi(x) = \frac{1}{\sqrt{2\pi}} \exp\left(-\frac{x^2}{2}\right).
\end{align*}
We recall that for this case, the gradient $H(\theta, \varepsilon)$ can be derived as
\begin{align*}
    H(\theta, \varepsilon) = - \frac{\pi^2(g_\theta(\varepsilon))}{q_\theta^2(g_\theta(\varepsilon))} \partial_\theta \log q_\theta(g_\theta(\varepsilon)),
\end{align*}
as $\partial_\theta g_\theta(\varepsilon) = 1$ and $Z_\pi = 1$. We can now discuss Assumptions~\ref{assump:LocalLip} and \ref{assmp:localDisp} for this case. We first note the general form
\begin{align*}
-\frac{\pi^2(x)}{q^2_\theta(x)} \partial_\theta \log q_\theta(x) = \frac{(1/2\pi) \exp(-x^2/2)}{(1/{\pi \nu \tau^2}) (1 + (x-\theta)^2/(\nu \tau^2))^{-\nu - 1}} \frac{(\nu + 1) (x - \theta)}{\tau^2 \nu + {(x - \theta)^2}}.
\end{align*}
Now, plugging $g_\theta(\varepsilon) = \tau \varepsilon + \theta$ into the above expression, we obtain
\begin{align}\label{eq:StochasticGradient_T}
H(\theta, \varepsilon) = -\frac{\pi^2(g_\theta(\varepsilon))}{q^2_\theta(g_\theta(\varepsilon))} \partial_\theta \log q_\theta(g_\theta(\varepsilon)) = \frac{(1/2\pi) \exp(-(\tau \varepsilon + \theta)^2)}{(1/{\pi \nu \tau^2}) (1 + \varepsilon^2/\nu)^{-\nu - 1}} \frac{(\nu + 1) \varepsilon}{\tau \nu + {\tau \varepsilon^2}}.
\end{align}
Inspecting the $H(\theta, \varepsilon)$ above, it is obvious that the function would not be uniformly Lipschitz in $\theta$ for any $\varepsilon$. The uniform Lipschitzness is a typical condition for optimization methods, which would be unjustified in this case. We show below that, our relaxation in Assumption~\ref{assump:LocalLip} is sufficient to cover this case. Specifically, we have the following proposition.
\begin{prop}\label{prop:LocalLip_example} Let $H(\theta, \varepsilon)$ be as defined in \eqref{eq:StochasticGradient_T}. Then, Assumption~\ref{assump:LocalLip} holds with $\rho = 7$, $L_1 = {\tau(\nu+1)}{\sqrt{1/2e}}$, and $L_2 = (\nu + 1) \nu \tau^2 \sqrt{1/2e}  + 4 (\nu+1) \nu \tau$.
\end{prop}
\begin{proof}
See Appendix~\ref{app:proof:prop:LocalLip_example}.
\end{proof}
A natural question is whether Assumption~\ref{assmp:localDisp} holds for this case. However, Assumption~\ref{assmp:localDisp} is a local dissipativity assumption, which is difficult to verify in general. In this particular example, Assumption~\ref{assmp:localDisp} turns out to be too restrictive. This intuitively means that the problem is not only \textit{non-convex} but it is even more ill-posed than our Assumption~\ref{assmp:localDisp} can handle. We remark nonetheless that our assumptions are the weakest known assumptions under which SGLD can be shown to converge \cite{zhang2023nonasymptotic}, therefore an improvement even for this simple example is challenging. We point out, however, there are two potential solutions to this kind of problem within our framework. First, one can resort to the kind of approach that is proposed in \cite{akyildiz2021convergence} for convexity results. One can simply assume that the parameter space is compact $\Theta \subset \bR^d$. Note that, this is relatively mild: As the parameters can be constrained to a space, it is a reasonable assumption to constrain the mean range of the proposal distribution (note that this is \textit{not} assuming the distributions would live on a compact set). This also eases the problem about Lipschitz constants for $\theta$ but the local Lipschitz property w.r.t. $\varepsilon$ would still need to be handled and Proposition~\ref{prop:LocalLip_example} provides a natural way. This assumption also necessitates including projection steps into the SGLD. Projected SGLD is known to stay close to SGLD (see, e.g., \cite{zou2021faster}), however, we are not aware of a full non-convex optimization result. Secondly, one can use approaches in \cite{lim2021non,lim2021polygonal} to weaken the dissipativity assumption. However, these works use a different discretization method -- hence we omit this and demonstrate below the performance of projected SGLD in this case.}

{For this particular example, we complement the reasoning above with numerical results which show that Assumption~\ref{assmp:localDisp} is stronger than necessary for convergence. We set up a challenging scenario by setting $\theta_0 = -10$. We choose $\Theta = [-10, 10]$ for this experiment, which is a reasonable assumption for the mean of the proposal distribution. We then implement the following scheme
\begin{align*}
    \theta_{k+1} = \mathsf{P}_\Theta \left( \theta_k - \eta H(\theta_k, \varepsilon_k) + \sqrt{2 \eta \beta^{-1}} \xi_{k+1} \right),
\end{align*}
where $\xi_{k} \sim \NPDF(0, 1)$, $\varepsilon_k \sim t_{\nu}(0, 1)$, and $\mathsf{P}_\Theta$ is the projection operator onto $\Theta$. Since the target Gaussian $\pi(x)$ is zero mean, we postulate that $\theta^\star = 0$. We fix $\tau = 1$ and $\nu = 1.5$. We note that, numerical evaluations of the gradient reveals that the gradient vanishes as $|\theta| \to \infty$ (which is one reason why dissipativity does not hold). We set $\beta = 5$ for algorithm to escape the flat valleys efficiently, $\eta = 10^{-2}$, and run the method for $T = 10^5$. We run $1000$ Monte Carlo runs and average the results. We plot the average curve in Figure~\ref{fig:sgld}. This result shows that projected SGLD is able to optimize the parameters of the importance sampling proposals when the proposal family is outside the exponential family.}
\section{Conclusions}\label{sec:conclusions}
We have provided global convergence rates for optimized adaptive importance samplers as introduced by \cite{akyildiz2021convergence}. Specifically, we considered the case of general proposal distributions and described adaptation schemes that globally optimize the $\chi^2$-divergence between the target and the proposal, leading to uniform error bounds for the resulting AIS schemes. Our approach is generic and can be adapted to several other schemes that are shown to be globally convergent. In other words, our guarantees apply when one replaces the SGLD with other optimizers, e.g., momentum-based (underdamped) optimizers \cite{akyildiz2020nonasymptotic}, variance reduced variants \cite{zou2019stochastic}, or tamed Euler schemes \cite{lim2021non} or polygonal schemes \cite{lim2021polygonal} which handle even more relaxed assumptions and enjoy improved stability.
\newpage
\appendix
\section*{Appendix}

{\section{Proofs}
\subsection{Proof of Proposition~\ref{prop:LocalLip_example}}\label{app:proof:prop:LocalLip_example}
We first note that
\begin{align*}
|H(\theta, \varepsilon) - H(\theta', \varepsilon)| &= \left| \frac{(1/2\pi) \exp(-(\tau \varepsilon + \theta)^2)}{(1/{\pi \nu \tau^2}) (1 + \varepsilon^2/\nu)^{-\nu - 1}} \frac{(\nu + 1) \varepsilon}{\tau \nu + {\tau \varepsilon^2}} - \frac{(1/2\pi) \exp(-(\tau \varepsilon + \theta')^2)}{(1/{\pi \nu \tau^2}) (1 + \varepsilon/\nu)^{-\nu - 1}} \frac{(\nu + 1) \varepsilon}{\tau \nu + {\tau \varepsilon^2}} \right|, \\
&= C_1(\varepsilon) \left| \exp(-(\tau \varepsilon + \theta)^2) - \exp(-(\tau \varepsilon + \theta')^2) \right|,
\end{align*}
where
\begin{align*}
    C_1(\varepsilon) =  \left| \frac{(1/2\pi)}{(1/{\pi \nu \tau^2}) (1 + \varepsilon^2/\nu)^{-\nu - 1}} \frac{(\nu + 1) \varepsilon}{\tau \nu + {\tau \varepsilon^2}} \right|,
\end{align*}
We first bound the last term using the fact that the function $x \mapsto \exp(-x^2)$ is Lipschitz with the constant $\sqrt{2/e}$. We obtain
\begin{align*}
\left| \exp(-(\tau \varepsilon + \theta)^2) - \exp(-(\tau \varepsilon + \theta')^2) \right| &\leq \sqrt{{2}/{e}} \left| (\tau \varepsilon + \theta) - (\tau \varepsilon + \theta') \right|, \\
&= \sqrt{{2}/{e}} |\theta - \theta'|.
\end{align*}
We now see the importance of allowing $\varepsilon$-dependence in Assumption~\ref{assump:LocalLip}(i). We can demonstrate that $C_1(\varepsilon)$ has the right dependence by noting
\begin{align*}
    C_1(\varepsilon) =  \left| \frac{(1/2\pi)}{(1/{\pi \nu \tau^2}) (1 + \varepsilon^2/\nu)^{-\nu - 1}} \frac{(\nu + 1) \varepsilon}{\tau \nu + {\tau \varepsilon^2}} \right| \leq \frac{\nu \tau^2}{2} \left(1 + \frac{\varepsilon^2}{\nu}\right)^{\nu + 1} \frac{(\nu + 1) \varepsilon}{\tau \nu},
\end{align*}
by using $\tau \nu + \tau \varepsilon^2 \geq \tau \nu$. We can now use the fact that $\nu > 1$ to obtain
\begin{align*}
    C_1(\varepsilon) \leq \frac{\tau (\nu+1)}{2} \left(1 + \varepsilon\right)^{2\nu + 3}.
\end{align*}
We see that the Assumption~\ref{assump:LocalLip}(i) holds with $\rho = 2 \nu + 3$ and $L_1 = {\tau(\nu+1)}{\sqrt{1/2e}}$. We now would like to check Assumption~\ref{assump:LocalLip}(ii), i.e., the Lipschitz continuity of $H(\theta, \varepsilon)$ in its second argument. We note that
\begin{align*}
    |H(\theta, \varepsilon) - H(\theta, \varepsilon')| &= \left| \frac{(1/2\pi) \exp(-(\tau \varepsilon + \theta)^2)}{(1/{\pi \nu \tau^2}) (1 + \varepsilon^2/\nu)^{-\nu - 1}} \frac{(\nu + 1) \varepsilon}{\tau \nu + {\tau \varepsilon^2}} - \frac{(1/2\pi) \exp(-(\tau \varepsilon' + \theta)^2)}{(1/{\pi \nu \tau^2}) (1 + \varepsilon'^2/\nu)^{-\nu - 1}} \frac{(\nu + 1) \varepsilon'}{\tau \nu + {\tau \varepsilon'^2}} \right|,
\end{align*}
In the interest of simplifying common constants, we rewrite that
\begin{align*}
    |H(\theta, \varepsilon) - H(\theta, \varepsilon')| &= \frac{(\nu+1) \nu \tau}{2} \left| \frac{\exp(-(\tau \varepsilon + \theta)^2)}{(1 + \varepsilon^2/\nu)^{-\nu - 1}} \frac{\varepsilon}{\nu + \varepsilon^2} - \frac{\exp(-(\tau \varepsilon' + \theta)^2)}{(1 + \varepsilon'^2/\nu)^{-\nu - 1}} \frac{\varepsilon'}{\nu + \varepsilon'^2} \right|.
\end{align*}
Let $c_1 = {(\nu+1) \nu \tau}/{2}$ for brevity. We start our estimates by a basic splitting to get rid of the exponential term
\begin{align*}
    |H(\theta, \varepsilon) - H(\theta, \varepsilon')| &\leq  c_1 \left| \frac{\exp(-(\tau \varepsilon + \theta)^2)}{(1 + \varepsilon^2/\nu)^{-\nu - 1}} \frac{\varepsilon}{\nu + \varepsilon^2} - \frac{\exp(-(\tau \varepsilon + \theta)^2)}{(1 + \varepsilon'^2/\nu)^{-\nu - 1}} \frac{\varepsilon'}{\nu + \varepsilon'^2} \right| \\
    &+ c_1 \left| \frac{\exp(-(\tau \varepsilon + \theta)^2)}{(1 + \varepsilon'^2/\nu)^{-\nu - 1}} \frac{\varepsilon'}{\nu + \varepsilon'^2} - \frac{\exp(-(\tau \varepsilon' + \theta)^2)}{(1 + \varepsilon'^2/\nu)^{-\nu - 1}} \frac{\varepsilon'}{\nu + \varepsilon'^2} \right|, \\
    &\leq c_1 \left| \frac{(1 + \varepsilon^2/\nu)^{+\nu + 1}\varepsilon}{\nu + \varepsilon^2} - \frac{(1 + \varepsilon'^2/\nu)^{+\nu + 1}\varepsilon'}{\nu + \varepsilon'^2} \right| \\
    &+ c_1 \left| \exp(-(\tau \varepsilon + \theta)^2) - \exp(-(\tau \varepsilon' + \theta)^2) \right| \left| \frac{\varepsilon' (1 + \varepsilon^2/\nu)^{\nu + 1}}{\nu + \varepsilon'^2} \right|.
\end{align*}
Using the fact that the function $x \mapsto \exp(-x^2)$ is $\sqrt{2/e}$-Lipschitz again, we get
\begin{align*}
    |H(\theta, \varepsilon) - H(\theta, \varepsilon')| &\leq  c_1 \left| \frac{(1 + \varepsilon^2/\nu)^{+\nu + 1}\varepsilon}{\nu + \varepsilon^2} - \frac{(1 + \varepsilon'^2/\nu)^{+\nu + 1}\varepsilon'}{\nu + \varepsilon'^2} \right| \\
    &+ c_1 \tau \sqrt{2/e} | \varepsilon - \varepsilon'| \left| \frac{\varepsilon' (1 + \varepsilon^2/\nu)^{\nu + 1}}{\nu + \varepsilon'^2} \right|.
\end{align*}
Some manipulations to the last term yields
\begin{align}\label{eq:stoc_grad_manipulations_1}
    |H(\theta, \varepsilon) - H(\theta, \varepsilon')| &\leq  c_1 \left| \frac{(1 + \varepsilon^2/\nu)^{+\nu + 1}\varepsilon}{\nu + \varepsilon^2} - \frac{(1 + \varepsilon'^2/\nu)^{+\nu + 1}\varepsilon'}{\nu + \varepsilon'^2} \right| + c_1 \tau \sqrt{2/e} | \varepsilon - \varepsilon'|  \left( 1 + |\varepsilon|\right)^{2\nu + 3}.
\end{align}
Now we turn our attention to the first term in equation~\eqref{eq:stoc_grad_manipulations_1} and note that (ignoring $c_1$ term for now):
\begin{align*}
    \left| \frac{(1 + \varepsilon^2/\nu)^{+\nu + 1}\varepsilon}{\nu + \varepsilon^2} - \frac{(1 + \varepsilon'^2/\nu)^{+\nu + 1}\varepsilon'}{\nu + \varepsilon'^2} \right| &\leq \left| \frac{(1 + \varepsilon^2/\nu)^{+\nu + 1}\varepsilon}{\nu + \varepsilon^2} - \frac{(1 + \varepsilon'^2/\nu)^{\nu + 1}\varepsilon'}{\nu + \varepsilon^2} \right| \\
    &+ \left| \frac{(1 + \varepsilon'^2/\nu)^{\nu + 1}\varepsilon'}{\nu + \varepsilon^2} - \frac{(1 + \varepsilon'^2/\nu)^{+\nu + 1}\varepsilon'}{\nu + \varepsilon'^2} \right|, \\
    &\leq \left| \varepsilon (1 + \varepsilon^2/\nu)^{\nu + 1} - \varepsilon' (1 + \varepsilon'^2/\nu)^{\nu + 1} \right| \\
    &+ \left|(1 + \varepsilon'^2/\nu)^{\nu + 1}\varepsilon'\right| \left| \frac{1}{\nu + \varepsilon^2} - \frac{1}{\nu + \varepsilon'^2} \right|.
\end{align*}
Note that the function $\varepsilon \mapsto 1/(\nu + \varepsilon)$ is $1$-Lipschitz, therefore,
\begin{align*}
    \left| \frac{(1 + \varepsilon^2/\nu)^{+\nu + 1}\varepsilon}{\nu + \varepsilon^2} - \frac{(1 + \varepsilon'^2/\nu)^{+\nu + 1}\varepsilon'}{\nu + \varepsilon'^2} \right| &\leq \left| \varepsilon (1 + \varepsilon^2/\nu)^{\nu + 1} - \varepsilon' (1 + \varepsilon'^2/\nu)^{\nu + 1} \right| + (1 + |\varepsilon'|)^{2\nu + 3} \left| \varepsilon - \varepsilon'\right|.
\end{align*}
Finally, we need to find the local Lipschitz constant of the first term above:
\begin{align*}
    \left| \varepsilon (1 + \varepsilon^2/\nu)^{\nu + 1} - \varepsilon' (1 + \varepsilon'^2/\nu)^{\nu + 1} \right| &\leq \left| \varepsilon (1 + \varepsilon^2/\nu)^{\nu + 1} - \varepsilon' (1 + \varepsilon^2/\nu)^{\nu + 1} \right| \\
    &+ \left| \varepsilon' (1 + \varepsilon^2/\nu)^{\nu + 1} - \varepsilon' (1 + \varepsilon'^2/\nu)^{\nu + 1} \right|, \\
    &\leq (1 + |\varepsilon|)^{2 \nu + 2} |\varepsilon - \varepsilon'| + |\varepsilon'| \left| (1 + \varepsilon^2/\nu)^{\nu + 1} - (1 + \varepsilon'^2/\nu)^{\nu + 1} \right|.
\end{align*}
We have to now identify the Lipschitz constant of the function $f(\varepsilon) = (1 + \varepsilon^2/\nu)^{\nu + 1}$. We note that
\begin{align*}
    f'(\varepsilon) = \frac{2 (\nu+1)}{\nu} (1 + \varepsilon^2/\nu)^\nu \varepsilon.
\end{align*}
We can see that this implies by the mean value theorem (for $\varepsilon < \varepsilon'$ without loss of generality) that
\begin{align*}
|f(\varepsilon) - f(\varepsilon')| \leq \sup_{c \in [\varepsilon, \varepsilon']} |f'(c)| |\varepsilon - \varepsilon'|.
\end{align*}
Hence, we will proceed to find an upper bound for $\sup_{c \in [\varepsilon, \varepsilon']} |f'(c)|$. Note that using $1 < \nu <2$, we can readily obtain
\begin{align*}
\sup_{c \in [\varepsilon, \varepsilon']} |f'(c)| &= \sup_{c \in [\varepsilon, \varepsilon']} \frac{2 (\nu+1)}{\nu} (1 + c^2/\nu)^\nu |c|, \\
&\leq \sup_{c \in [\varepsilon,\varepsilon']} 6 (1 + c^2)^2 |c|, \\
&\leq \sup_{c\in [\varepsilon, \varepsilon']} 6 (1 + |c|)^5, \\
&\leq 6 (1 + |\varepsilon| + |\varepsilon'|)^5.
\end{align*}
The last line is true regardless of the ordering of $\varepsilon$ and $\varepsilon'$. We can now conclude that
\begin{align*}
    |\varepsilon'| \left| (1 + \varepsilon^2/\nu)^{\nu + 1} - (1 + \varepsilon'^2/\nu)^{\nu + 1} \right| \leq 6 (1 + |\varepsilon| + |\varepsilon'|)^6 |\varepsilon - \varepsilon'|.
\end{align*}
Finally, merging all these, we can upper bound the first term of
\begin{align*}
|H(\theta, \varepsilon) - H(\theta, \varepsilon')| &\leq (c_1 \tau \sqrt{2/e} + 8 c_1) \left(1 + |\varepsilon| + |\varepsilon'|\right)^7 |\varepsilon - \varepsilon'|,
\end{align*}
which shows that Assumption~\ref{assump:LocalLip}(ii) holds with $L_2 = (c_1 \tau \sqrt{2/e} + 8 c_1)$ and $\rho = 7$. This concludes the proof. \qed}
\bibliographystyle{amsplain}
\bibliography{../draft}

\providecommand{\bysame}{\leavevmode\hbox to3em{\hrulefill}\thinspace}
\providecommand{\MR}{\relax\ifhmode\unskip\space\fi MR }
\providecommand{\MRhref}[2]{%
  \href{http://www.ams.org/mathscinet-getitem?mr=#1}{#2}
}
\providecommand{\href}[2]{#2}
\begin{thebibliography}{10}

\bibitem{agapiou2017importance}
S~Agapiou, Omiros Papaspiliopoulos, D~Sanz-Alonso, and AM~Stuart,
  \emph{Importance sampling: Intrinsic dimension and computational cost},
  Statistical Science \textbf{32} (2017), no.~3, 405--431.

\bibitem{akyildiz2021convergence}
{\"O}mer~Deniz Akyildiz and Joaqu{\'\i}n M{\'\i}guez, \emph{Convergence rates
  for optimised adaptive importance samplers}, Statistics and Computing
  \textbf{31} (2021), no.~2, 1--17.

\bibitem{akyildiz2020nonasymptotic}
{\"O}mer~Deniz Akyildiz and Sotirios Sabanis, \emph{Nonasymptotic analysis of
  {S}tochastic {G}radient {H}amiltonian {M}onte {C}arlo under local conditions
  for nonconvex optimization}, arXiv preprint arXiv:2002.05465 (2020).

\bibitem{arouna2004adaptative}
Bouhari Arouna, \emph{Adaptative monte carlo method, a variance reduction
  technique}, Monte Carlo Methods and Applications \textbf{10} (2004), no.~1,
  1--24.

\bibitem{arouna2004robbins}
\bysame, \emph{Robbins-{M}onro algorithms and variance reduction in finance},
  Journal of Computational Finance \textbf{7} (2004), no.~2, 35--62.

\bibitem{bengio2008adaptive}
Yoshua Bengio and Jean-S{\'e}bastien Sen{\'e}cal, \emph{Adaptive importance
  sampling to accelerate training of a neural probabilistic language model},
  IEEE Transactions on Neural Networks \textbf{19} (2008), no.~4, 713--722.

\bibitem{brosse2018promises}
Nicolas Brosse, Alain Durmus, and Eric Moulines, \emph{The promises and
  pitfalls of stochastic gradient langevin dynamics}, Advances in Neural
  Information Processing Systems \textbf{31} (2018).

\bibitem{bugallo2017adaptive}
Monica~F Bugallo, Victor Elvira, Luca Martino, David Luengo, Joaquin Miguez,
  and Petar~M Djuric, \emph{{A}daptive {I}mportance {S}ampling: {T}he past, the
  present, and the future}, IEEE Signal Processing Magazine \textbf{34} (2017),
  no.~4, 60--79.

\bibitem{bugallo2015adaptive}
M{\'o}nica~F Bugallo, Luca Martino, and Jukka Corander, \emph{Adaptive
  importance sampling in signal processing}, Digital Signal Processing
  \textbf{47} (2015), 36--49.

\bibitem{cappe2008adaptive}
Olivier Capp{\'e}, Randal Douc, Arnaud Guillin, Jean-Michel Marin, and
  Christian~P Robert, \emph{Adaptive importance sampling in general mixture
  classes}, Statistics and Computing \textbf{18} (2008), no.~4, 447--459.

\bibitem{cappe2004population}
Olivier Capp{\'e}, Arnaud Guillin, Jean-Michel Marin, and Christian~P Robert,
  \emph{Population {M}onte {C}arlo}, Journal of Computational and Graphical
  Statistics \textbf{13} (2004), no.~4, 907--929.

\bibitem{chau2021stochastic}
Ngoc~Huy Chau, {\'E}ric Moulines, Miklos R{\'a}sonyi, Sotirios Sabanis, and
  Ying Zhang, \emph{On stochastic gradient langevin dynamics with dependent
  data streams: The fully nonconvex case}, SIAM Journal on Mathematics of Data
  Science \textbf{3} (2021), no.~3, 959--986.

\bibitem{cornebise2008adaptive}
Julien Cornebise, {\'E}ric Moulines, and Jimmy Olsson, \emph{Adaptive methods
  for sequential importance sampling with application to state space models},
  Statistics and Computing \textbf{18} (2008), no.~4, 461--480.

\bibitem{delyon2018asymptotic}
Bernard Delyon and Fran{\c{c}}ois Portier, \emph{Asymptotic optimality of
  adaptive importance sampling}, Proceedings of the 32nd International
  Conference on Neural Information Processing Systems, 2018, pp.~3138--3148.

\bibitem{dieng2017variational}
Adji~Bousso Dieng, Dustin Tran, Rajesh Ranganath, John Paisley, and David Blei,
  \emph{Variational inference via $\chi$-upper bound minimization}, Advances in
  Neural Information Processing Systems, 2017, pp.~2732--2741.

\bibitem{douc2007convergence}
Randal Douc, Arnaud Guillin, J-M Marin, and Christian~P Robert,
  \emph{Convergence of adaptive mixtures of importance sampling schemes}, The
  Annals of Statistics \textbf{35} (2007), no.~1, 420--448.

\bibitem{el2019stochastic}
Yousef El-Laham and M{\'o}nica~F Bugallo, \emph{Stochastic gradient population
  monte carlo}, IEEE Signal Processing Letters \textbf{27} (2019), 46--50.

\bibitem{el2021policy}
\bysame, \emph{Policy gradient importance sampling for bayesian inference},
  IEEE Transactions on Signal Processing \textbf{69} (2021), 4245--4256.

\bibitem{el2019variational}
Yousef El-Laham, Petar~M Djuri{\'c}, and M{\'o}nica~F Bugallo, \emph{A
  variational adaptive population importance sampler}, ICASSP 2019-2019 IEEE
  International Conference on Acoustics, Speech and Signal Processing (ICASSP),
  IEEE, 2019, pp.~5052--5056.

\bibitem{elvira2019langevin}
V{\'\i}ctor Elvira and {\'E}milie Chouzenoux, \emph{Langevin-based strategy for
  efficient proposal adaptation in population monte carlo}, ICASSP 2019-2019
  IEEE International Conference on Acoustics, Speech and Signal Processing
  (ICASSP), IEEE, 2019, pp.~5077--5081.

\bibitem{elvira2021optimized}
V{\'\i}ctor Elvira and Emilie Chouzenoux, \emph{Optimized population monte
  carlo},  (2021).

\bibitem{elvira2022gradient}
V{\'\i}ctor Elvira, Emilie Chouzenoux, {\"O}mer~Deniz Akyildiz, and Luca
  Martino, \emph{Gradient-based adaptive importance samplers}, arXiv preprint
  arXiv:2210.10785 (2022).

\bibitem{elvira2017improving}
V{\'\i}ctor Elvira, Luca Martino, David Luengo, and M{\'o}nica~F Bugallo,
  \emph{Improving population monte carlo: Alternative weighting and resampling
  schemes}, Signal Processing \textbf{131} (2017), 77--91.

\bibitem{elvira2019generalized}
V{\'\i}ctor Elvira, Luca Martino, David Luengo, M{\'o}nica~F Bugallo, et~al.,
  \emph{Generalized multiple importance sampling}, Statistical Science
  \textbf{34} (2019), no.~1, 129--155.

\bibitem{elvira2022rethinking}
V{\'\i}ctor Elvira, Luca Martino, and Christian~P Robert, \emph{Rethinking the
  effective sample size}, International Statistical Review \textbf{90} (2022),
  no.~3, 525--550.

\bibitem{erdogdu2018global}
Murat~A Erdogdu, Lester Mackey, and Ohad Shamir, \emph{Global non-convex
  optimization with discretized diffusions}, arXiv preprint arXiv:1810.12361
  (2018).

\bibitem{fasiolo2018langevin}
Matteo Fasiolo, Fl{\'a}vio~Eler de~Melo, and Simon Maskell, \emph{Langevin
  incremental mixture importance sampling}, Statistics and Computing
  \textbf{28} (2018), no.~3, 549--561.

\bibitem{fu2002optimal}
MC~Fu and Y~Su, \emph{Optimal importance sampling in securities pricing},
  Journal of Computational Finance \textbf{5} (2002), no.~4, 27--50.

\bibitem{gao2021global}
Xuefeng Gao, Mert G{\"u}rb{\"u}zbalaban, and Lingjiong Zhu, \emph{Global
  convergence of stochastic gradient hamiltonian monte carlo for nonconvex
  stochastic optimization: Nonasymptotic performance bounds and momentum-based
  acceleration}, Operations Research (2021).

\bibitem{hwang1980laplace}
Chii-Ruey Hwang, \emph{Laplace's method revisited: weak convergence of
  probability measures}, The Annals of Probability (1980), 1177--1182.

\bibitem{jerfel2021variational}
Ghassen Jerfel, Serena Wang, Clara Wong-Fannjiang, Katherine~A Heller, Yian Ma,
  and Michael~I Jordan, \emph{Variational refinement for importance sampling
  using the forward kullback-leibler divergence}, Uncertainty in Artificial
  Intelligence, PMLR, 2021, pp.~1819--1829.

\bibitem{kappen2016adaptive}
Hilbert~Johan Kappen and Hans~Christian Ruiz, \emph{Adaptive importance
  sampling for control and inference}, Journal of Statistical Physics
  \textbf{162} (2016), no.~5, 1244--1266.

\bibitem{kawai2008adaptive}
Reiichiro Kawai, \emph{Adaptive monte carlo variance reduction for l{\'e}vy
  processes with two-time-scale stochastic approximation}, Methodology and
  Computing in Applied Probability \textbf{10} (2008), no.~2, 199--223.

\bibitem{kawai2017acceleration}
\bysame, \emph{Acceleration on adaptive importance sampling with sample average
  approximation}, SIAM Journal on Scientific Computing \textbf{39} (2017),
  no.~4, A1586--A1615.

\bibitem{kawai2018optimizing}
\bysame, \emph{Optimizing adaptive importance sampling by stochastic
  approximation}, SIAM Journal on Scientific Computing \textbf{40} (2018),
  no.~4, A2774--A2800.

\bibitem{lapeyre2011framework}
Bernard Lapeyre and J{\'e}r{\^o}me Lelong, \emph{A framework for adaptive monte
  carlo procedures}, Monte Carlo Methods and Applications \textbf{17} (2011),
  no.~1, 77--98.

\bibitem{lim2021non}
Dong-Young Lim, Ariel Neufeld, Sotirios Sabanis, and Ying Zhang,
  \emph{Non-asymptotic estimates for tusla algorithm for non-convex learning
  with applications to neural networks with relu activation function}, arXiv
  preprint arXiv:2107.08649 (2021).

\bibitem{lim2021polygonal}
Dong-Young Lim and Sotirios Sabanis, \emph{Polygonal unadjusted langevin
  algorithms: Creating stable and efficient adaptive algorithms for neural
  networks}, arXiv preprint arXiv:2105.13937 (2021).

\bibitem{llorente2021mcmc}
Fernando Llorente, E~Curbelo, Luca Martino, Victor Elvira, and D~Delgado,
  \emph{{MCMC}-driven importance samplers}, arXiv preprint arXiv:2105.02579
  (2021).

\bibitem{lopez2020decision}
Romain Lopez, Pierre Boyeau, Nir Yosef, Michael Jordan, and Jeffrey Regier,
  \emph{Decision-making with auto-encoding variational bayes}, Advances in
  Neural Information Processing Systems \textbf{33} (2020).

\bibitem{martino2017anti}
Luca Martino, Victor Elvira, and David Luengo, \emph{Anti-tempered layered
  adaptive importance sampling}, 2017 22nd International Conference on Digital
  Signal Processing (DSP), IEEE, 2017, pp.~1--5.

\bibitem{martino2015adaptive}
Luca Martino, Victor Elvira, David Luengo, and Jukka Corander, \emph{An
  adaptive population importance sampler: Learning from uncertainty}, IEEE
  Transactions on Signal Processing \textbf{63} (2015), no.~16, 4422--4437.

\bibitem{martino2017layered}
\bysame, \emph{Layered adaptive importance sampling}, Statistics and Computing
  \textbf{27} (2017), no.~3, 599--623.

\bibitem{mousavi2021hamiltonian}
Ali Mousavi, Reza Monsefi, and V{\'\i}ctor Elvira, \emph{Hamiltonian adaptive
  importance sampling}, IEEE Signal Processing Letters \textbf{28} (2021),
  713--717.

\bibitem{muller2019neural}
Thomas M{\"u}ller, Brian McWilliams, Fabrice Rousselle, Markus Gross, and Jan
  Nov{\'a}k, \emph{Neural importance sampling}, ACM Transactions on Graphics
  (ToG) \textbf{38} (2019), no.~5, 1--19.

\bibitem{perello2023adaptively}
Carlos~ACC Perello and {\"O}mer~Deniz Akyildiz, \emph{Adaptively optimised
  adaptive importance samplers}, arXiv preprint arXiv:2307.09341 (2023).

\bibitem{pradier2019challenges}
Melanie~F Pradier, Michael~C Hughes, and Finale Doshi-Velez, \emph{Challenges
  in computing and optimizing upper bounds of marginal likelihood based on
  chi-square divergences}, Symposium on Advances in Approximate Bayesian
  Inference (2019).

\bibitem{raginsky2017non}
Maxim Raginsky, Alexander Rakhlin, and Matus Telgarsky, \emph{Non-convex
  learning via stochastic gradient {L}angevin dynamics: a nonasymptotic
  analysis}, Conference on Learning Theory, 2017, pp.~1674--1703.

\bibitem{ryu2016convex}
Ernest~K Ryu, \emph{Convex optimization for monte carlo: Stochastic
  optimization for importance sampling}, Ph.D. thesis, Stanford University,
  2016.

\bibitem{ryu2014adaptive}
Ernest~K Ryu and Stephen~P Boyd, \emph{Adaptive importance sampling via
  stochastic convex programming}, arXiv:1412.4845 (2014).

\bibitem{sanz2018importance}
Daniel Sanz-Alonso, \emph{Importance sampling and necessary sample size: an
  information theory approach}, SIAM/ASA Journal on Uncertainty Quantification
  \textbf{6} (2018), no.~2, 867--879.

\bibitem{sanz2021bayesian}
Daniel Sanz-Alonso and Zijian Wang, \emph{Bayesian update with importance
  sampling: Required sample size}, Entropy \textbf{23} (2021), no.~1, 22.

\bibitem{welling2011bayesian}
Max Welling and Yee~W Teh, \emph{Bayesian learning via stochastic gradient
  langevin dynamics}, Proceedings of the 28th international conference on
  machine learning (ICML-11), 2011, pp.~681--688.

\bibitem{xu2019variance}
Ming Xu, Matias Quiroz, Robert Kohn, and Scott~A Sisson, \emph{Variance
  reduction properties of the reparameterization trick}, The 22nd international
  conference on artificial intelligence and statistics, PMLR, 2019,
  pp.~2711--2720.

\bibitem{xu2018global}
Pan Xu, Jinghui Chen, Difan Zou, and Quanquan Gu, \emph{Global convergence of
  langevin dynamics based algorithms for nonconvex optimization}, Advances in
  {N}eural {I}nformation {P}rocessing {S}ystems (NeurIPS) (2018).

\bibitem{zhang2023nonasymptotic}
Ying Zhang, {\"O}mer~Deniz Akyildiz, Theodoros Damoulas, and Sotirios Sabanis,
  \emph{Nonasymptotic estimates for stochastic gradient langevin dynamics under
  local conditions in nonconvex optimization}, Applied Mathematics \&
  Optimization \textbf{87} (2023), no.~2, 25.

\bibitem{zou2019stochastic}
Difan Zou, Pan Xu, and Quanquan Gu, \emph{Stochastic gradient {H}amiltonian
  {M}onte {C}arlo methods with recursive variance reduction}, Advances in
  Neural Information Processing Systems \textbf{32} (2019), 3835--3846.

\bibitem{zou2021faster}
\bysame, \emph{Faster convergence of stochastic gradient langevin dynamics for
  non-log-concave sampling}, Uncertainty in Artificial Intelligence, PMLR,
  2021, pp.~1152--1162.

\end{thebibliography}

\end{document}